\documentclass[journal]{IEEEtran}
\usepackage{pifont}
\usepackage{}
\usepackage{mathrsfs}
\usepackage[noadjust]{cite}
\usepackage{graphicx}
\usepackage{float}
\usepackage{subfigure}
\usepackage{amsthm}
\usepackage{amsmath}
\usepackage{bm}
\usepackage{enumerate}
\usepackage{amssymb}
\usepackage{fixltx2e}
\usepackage{color}
\usepackage{algorithm}
\usepackage{multirow}
\usepackage{ragged2e}
\usepackage{algpseudocode}
\usepackage{graphics}
\usepackage{epsfig}

\newtheorem{lemma}{Lemma}

\allowdisplaybreaks[4]
\begin{document}

\title{Hybrid Multicast/Unicast Design\\ in NOMA-based Vehicular Caching System with Supplementary Material}
\author{Xinyue~Pei, Hua~Yu,~\IEEEmembership{Member,~IEEE,}
		Yingyang~Chen,~\IEEEmembership{Member,~IEEE,}\\ 
		Miaowen~Wen,~\IEEEmembership{Senior Member,~IEEE,}
		and Gaojie Chen,~\IEEEmembership{Senior Member,~IEEE}
		 
		\thanks{X. Pei, M. Wen, and H. Yu are with the National Engineering Technology Research Center for Mobile Ultrasonic Detection, South China University of Technology, Guangzhou 510640, China (e-mail:
		eexypei@mail.scut.edu.cn;
		\{eemwwen, yuhua\}@scut.edu.cn).}
	\thanks{Y. Chen is with Department of Electronic Engineering, 
		College of Information Science and Technology, Jinan University, Guangzhou (e-mail:
		chenyy@jnu.edu.cn).}
	\thanks{G. Chen is with
		the School of Engineering, University of Leicester, Leicester LE1 7HB,
		U.K. (e-mail:  gaojie.chen@leicester.ac.uk).}
}
\markboth{SUBMITTED TO TRANSACTIONS ON VEHICULAR TECHNOLOGY}{}
\maketitle

\begin{abstract}
  In this paper,  we investigate a hybrid multicast/unicast scheme for a multiple-input single-output cache-aided non-orthogonal multiple access (NOMA) vehicular  
	scenario in the face of rapidly fluctuating
	vehicular wireless channels. Considering a more practical situation,  imperfect channel state information is  taking   into account. In this paper, we formulate an
	optimization problem  to maximize the unicast sum rate under the constraints of the peak power, the peak backhaul, the minimum unicast rate, and the maximum multicast outage probability. To solve the formulated non-convex problem, a lower bound relaxation method is proposed, which enables a division of the original problem into two convex sub-problems. Computer simulations show that
	the proposed caching-aided NOMA is superior to the orthogonal multiple
	access counterpart.
\end{abstract}

\maketitle
\section{Introduction}
Recently, multicast services have been gaining huge interest  in cellular networks [\ref{multicast}]. With the increasing demand of accessing to both multicast (e.g., proactive content pushing) and unicast services (e.g., targeted advertisements), the hybrid design of multicast and unicast services is a hot topic in the next-generation wireless communication studies [\ref{multicastnoma}]. According to the standards ratified by the  3rd generation partnership project
(3GPP), multicast and unicast services need to be divided into different time slots or frequencies [\ref{3gpp}], [\ref{refDW}]. On the other hand, non-orthogonal multiple access (NOMA) is a recognized next-generation technology, which shows superior
spectral efficiency performance compared to conventional orthogonal multiple
access (OMA)  [\ref{ratedefine}], [\ref{lldnoma}]. Unlike OMA, NOMA can distinguish users in the power domain by using successive
interference cancellation (SIC) techniques.  Compared to conventional cellular networks (e.g., LTE-multicast [\ref{3gpp}]),  NOMA-based hybrid design can realize the requirements in the power-domain.  Therefore, applying the NOMA technique to the design of a hybrid multicast/unicast system is envisioned to improve the efficiency of the system significantly [\ref{multicastnoma}]. 

The internet-of-vehicles ecosystem is another crucial technique in the future, in which vehicles need to exchange a massive amount of data with the cloud, resulting in substantial backhaul overhead [\ref{IOV}]. As a result, wireless edge caching technology is envisioned to resolve this challenge by storing contents at edge users or base stations in advance during off-peak time  [\ref{nomacaching}], [\ref{taocaching}]. To further enhance system performance for vehicular communication, NOMA is applied [\ref{17}], [\ref{NOMAvehicle}]. Therefore, it is clearly that the combination of caching, NOMA, and vehicular system is feasible and promising. Nevertheless, to the best of our knowledge, only one work [\ref{cachenomavehicular}] investigates a two-user cache-aided NOMA vehicular  network. However, the users' mobility and multiple receivers have not been taken into consideration. 

 In this context, we introduce a cache-aided NOMA vehicular  scheme for a hybrid multicast/unicast system with a backhaul-capacity constraint in the face of rapidly fluctuating
 vehicular wireless channels. Without loss of generality, we  consider one multicast user cluster and $K$  unicast users with high mobility. Additionally, we consider  the imperfect Gaussian-distributed channel state information (CSI). The main contributions of this paper are summarized below: 
\begin{itemize}
	\item  We study a generalized and practical cache-aided NOMA vehicular system, where $K$ high-speed unicast vehicular users and one multicast user cluster coexist. Moreover, we take the backhaul constraint and imperfect CSI (I-CSI) into consideration and study their impacts on the proposed schemes. 
	\item  We formulate an optimization problem for the joint
	design   in order to find the maximum sum rate of unicast users. With the aid of a proposed lower bound relaxation method, we turn the non-convex problem into a convex problem. We achieve a feasible solution by  dividing the formulated problem into two convex sub-problems.
	\item We compare the cache-aided NOMA scheme with the cache-aided OMA one. Results reveal that the NOMA scheme achieves a much higher unicast sum rate than the OMA scheme. In addition, it shows that the cache-aided system can alleviate the backhaul link.\footnote{\emph{Notation}:  $\mathcal N_c(\mu,\sigma_0^2)$ denotes complex Gaussion distribution with mean $\mu$ and variance $\sigma_0^2$. ${F}_X(\cdot)$  denotes the  cumulative distribution function (CDF)  of random variable  $X$.}
\end{itemize}



  \begin{figure}[t]
	\centering
	\includegraphics[width=3.5in]{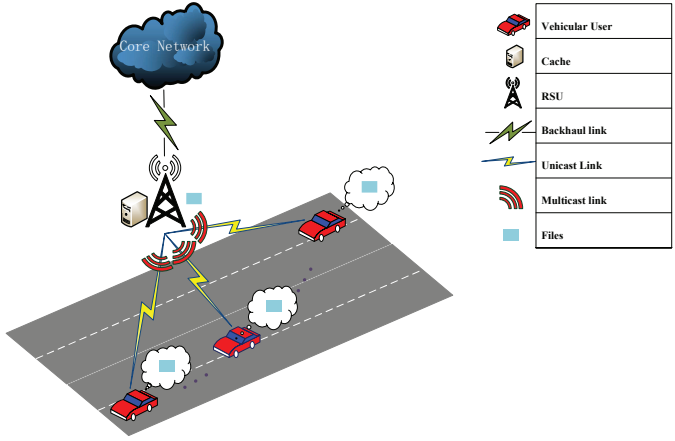}
	\caption{System model.} 	
	\label{system model}
\end{figure}
\section{System Model}
We consider a vehicular downlink single-input single-output (SISO) transmission system,  where
a roadside unit (RSU), configured with one transmit antenna, 
provides hybrid multicast and unicast services to $K$ vehicular users (denoted by $U_i$, $i\in\{1,...,K\})$, equipped with a single antenna. As shown in Fig. 1, RSU is allocated with some cache resources, and the backhaul link of RSU is assumed to be capacity-limited. For simplicity of analysis, we study the case of a single one multicast group, i.e., $\left\{U_i\right\}$, while the case of multiple multicast groups will be extended in the future work. 

\subsection{Transmission Model}
Let ${x}_M$, ${x}_i=$ ($i\in\{1,...,K\}$) be the data symbols corresponding to  multicast and unicast transmissions, respectively. All the
data symbols are assumed to have the same unit power, i.e., $\text{E}[\vert x_M\vert^2]=\text{E}[\vert x_i\vert^2]=1$.  It is assumed that
RSU uses the NOMA protocol to send the superimposed signal
to all users, which apply the SIC technique to decode the signal. To be realistic, we assume that the channel estimation processes are
imperfect [\ref{channelmodel}]. Hence, we have 
${h}_i(t)=\sqrt{1-\phi^2}\hat{{{h}}}_i(t)+\phi {\epsilon}_i(t),$
where ${h}_i$ denotes  the channel vector from RSU to $U_i$, $\hat{{{h}}}_i(t)\in\mathbb{C}^{Nt\times 1}$  denotes  the estimated channel vector between the same nodes with variance $\Omega_{i}$, and ${\epsilon}_i\in\mathbb{C}^{Nt\times 1}$ denotes the estimation error vector with variance $\Omega_{\epsilon,i}$. For convenience,  $\Omega_{\epsilon,i}$ is assumed to be a constant $\Omega_\epsilon$. All the channels are characterized by Jakes'  model [\ref{fadingchannel}] to measure users' mobility, i.e.,  $\phi=J_0(2\pi f_cv_i/c\tau),$  where  $J_0$ is the zeroth-order Bessel function of the first kind, $f_c$ denotes the carrier frequency, $v_i$ indicates the moving
velocity of $U_i$, $c$ is the light speed, and $\tau$   represents the duration between two adjacent time slots. Without loss of generality, we sort the average power of ${\rm RSU}-U_i$ links as $\vert{h}_1\vert^2\geq\cdots\geq\vert{h}_K\vert^2$. For the sake of fairness, a minimum rate limitation, namely $r_{min}$ is set. The unicast rate of each user must satisfy $r_i^U\geq r_{min}$. 

Considering any active time slot, RSU transmits a superimposed signal as
${z}=\sqrt{\beta_MP}{x}_M+\sum_{i=1}^{K}\sqrt{\beta_UP_i}{x}_i,$
where $\beta_M$ and $\beta_U$  denote the power allocation coefficients for  multicast and unicast transmissions,  respectively; $P$ and $P_i$ $(\sum_{i=1}^{K}P_i=P)$ denote the transmit power in multicast layer and for $U_i$ in unicast layer, respectively. Let ${y}_i$ be the received signal at $U_i$,
which is given as:
${y}_i=\sqrt{\beta_MP}{h}_i{x}_M+\sum_{j=1}^{K}\sqrt{\beta_UP_j}{h}_i{x}_j+n_i,$
where $n_i\sim\mathcal N_c(0,\Omega_0)$ is additive white
Gaussian noise (AWGN). Because in the downlink
system, multicast mode is more resource-efficient
than unicast mode, 
multicast messages should
have a higher priority [\ref{asurvey}]. Therefore, the multicast
messages are assumed to be decoded and subtracted before decoding
the unicast messages. Thus, the
data rate of ${x}_M$ at $U_i$ can be obtained as
\begin{align}\label{rim}
r_i^M=\log_2\left(1+\frac{\rho_M\lambda_i}{\rho_U\lambda_i+\Psi}\right),
\end{align}
where $\lambda_i=\vert\hat{{{h}}}_{i(n)}\vert^2$, $\rho_M=\beta_MP/\Omega_0$, $\rho_U=\beta_UP/\Omega_0$, $a=1/(1-\phi^2)$, $b=\phi^2/(1-\phi^2)\Omega_\epsilon$, and $\Psi=(\rho_M+\rho_U)b+a$. Obviously,  $\rho_M+\rho_U=\rho$, where $\rho=P/\Omega_0$. Similarly, the  instantaneous rate of $x_i$  observed at $U_i$ can be derived as
\begin{align}\label{riU}
r_i^U=\log_2\left(1+\rho_i\lambda_i/(\sum_{j=1}^{i-1}\rho_j\lambda_i+\sum_{j=1}^{i}\rho_jb+a)\right),
\end{align}
for $i\in\left\{1,...,K\right\}$, 
where  $\rho_i=\beta_UP_i/\Omega_0$. The detailed derivations of (\ref{rim}) and (\ref{riU}) are shown in the end of this paper.

\subsection{Cache Model}
We assume that  the ergodic rate of the backhaul link between RSU and the core
network
is subject to $R$ bit/s/Hz. Besides, we assume that RSU is equipped with a finite capacity
cache of size $N$. Let $\mathscr{F} = \{1, 2,\cdots, F\}$ denote the content of $F$ files,
each with normalized size of 1. Obviously, 
not all users can ask for their unicast messages at a time slot.  As adopted in most existing works [\ref{zipf}],  the popularity profile on $\mathscr{F}$ is 
modeled by a Zipf distribution, with a skewness
control parameter $\zeta$. Specifically, the popularity of file $f$ (denoted by $q_f$, $f\in \mathscr{F}$), is given by   $q_f={f^{-\zeta}}/{\sum_{j=1}^{F}j^{-\zeta}},$
which follows $\sum_{f=1}^{F}q_f=1$. Let $c_f$ represent the  probability that RSU caches the file $f$, satisfying $0\leq c_f\leq1$. Due to cache capacity limit at RSU, we can obtain $\sum_{f=1}^{F}c_f\leq N$.

\section{Problem Formulation}
Without loss of generality, the reception performance of multicast messages ${x}_M$ should meet the users’ quality of service (QoS) requirements, i.e., each user has a preset target rate $R_M$. As for unicast messages, they are assumed to be received opportunistically according to the
user’s channel condition [\ref{ratedefine}]. Therefore, we use the outage probabilities and instantaneous achievable rates to measure the reception performance of multicast and unicast messages, respectively.
\subsection{Outage  Probability}
Since the CDF of $\lambda_i$ is $F_{\lambda_i}(x)=1-\exp(-x/\Omega_{i})$, given the definition of the  outage
probability of ${x}_M$ at $U_i$ (denoted by $P_i^M$), namely, $P_i^M=\Pr\{r_i^M<R_M\}$, we have
\begin{align}\label{P0i}
P_i^M&=1-\exp\left(-\frac{\Psi\theta_M}{(\rho_M-\theta_M\rho_U)\Omega_i}\right),
\end{align}
where $\theta_M=2^{R_M}-1$. Obviously, $P_i^M>0$; in other words, we have $\rho_U<\rho/2^{R_M}$.\footnote{This condition ensures the quality of service of multicast signals, but may not hold when there exists interrupt.}

\subsection{Optimization Problem}
Notably, our objective is to maximize the sum rate of unicast signals, and
the optimization problem can be formulated as
\begin{subequations}
	\begin{align}
	P_0:\quad \mathop {\max\limits_{c_f,\rho_U,\rho_i}}&\quad \sum_{i=1}^{K}r_{i}^U\nonumber\\
	s.t.
	\label{4a}\quad& P_i^M<\delta,\\
	\label{4b}\quad& r_i^U\geq r_{min},\\
	\label{4c}\quad&\sum_{i=1}^{K}\rho_i=\rho_U,\\
	\label{4d}\quad&\rho_M+\rho_U=\rho,\\
	\label{4e}\quad&\sum_{f=1}^{F}\sum_{i=1}^{K}q_f(1-c_f)r_{i}^U\leq R,\\
	\label{4f}\quad&0\leq c_f\leq 1,\\
	\label{4g}\quad&\sum_{f=1}^{F}c_f\leq N,
	\end{align}
\end{subequations}
where (\ref{4a}) and (\ref{4b}) indicate the QoS requirements for the multicast and the unicast messages, respectively; (\ref{4c}) and (\ref{4d}) denote transmit power relationships for different signals;  (\ref{4e}) indicates the backhaul capacity constraint;  (\ref{4f}) indicates  the value range of cache probability; (\ref{4g}) represents the cache capacity limit at the RSU.
Without loss of generality, we have the outage requirement $\delta$ satisfying $0<\delta<1$, i.e., $\ln(1-\delta)<0$. Therefore, by substituting (\ref{P0i}) and (4d) into (4a), for  $\rho_U<\rho/2^{R_M}$, we can arrive at $\rho_U\leq {\Psi\theta_M}/{(2^{R_M}\Omega_i\ln(1-\delta))}+{\rho}/{2^{R_M}}.$ Therefore, $P_0$ can be equivalently rewritten as
\begin{subequations}\label{P1}
	\begin{align}
	P_1:\quad \mathop {\max\limits_{c_f,\rho_U,\rho_i}}&\quad \sum_{i=1}^{K}r_{i}^U\nonumber\\
	s.t.
	\label{5a}\quad& \rho_U\leq \frac{\Psi\theta_M}{2^{R_M}\Omega_1\ln(1-\delta)}+\frac{\rho}{2^{R_M}},\\
	\quad&(\ref{4b} ),(\ref{4c} ),(\ref{4e} )-(\ref{4g}).\nonumber
	\end{align}
\end{subequations}

\section{Proposed Lower Bound Relaxation Method}
Evidently, the objective function of $P_1$ is non-convex and hard to solve. Moreover, as shown in  (\ref{riU}), $(\lambda_i+b)$ in the denominator makes $r_i^U$ hardly be reformulated. Therefore, we use the \emph{lower bound relaxation} method, which can be derived as
\begin{align}\label{1}
r_i^U=\log_2\left(1+\rho_i\lambda_i/(\sum_{j=1}^{i-1}\rho_j\lambda_i+\Psi)\right).
\end{align} 
The detailed derivation of (\ref{1}) is shown in the end of this paper.
Invoking [\ref{chennoma}],
  $\rho_U$ can be split into two parts: $\rho_{min}$ for $r_{min}$ and $\triangle \rho$ for $\sum_{i=1}^{K}\triangle r_i^U$. The minimum transmit signal-to-noise ratio (SNR) and the excess transmit SNR  of $U_i$ are denoted by $\rho_{i,min}$ and $\triangle\rho_i$, respectively.\footnote{It is assumed that with $\rho_{i,min}$ ($i\in\{1,...,K\}$), $U_i$ will achieve the same data rate $r_{min}$. } Apparently, we have $\rho_{min}=\sum_{i=1}^{K}\rho_{i,min}$ and $\triangle\rho=\sum_{i=1}^{K}\triangle\rho_i$. For convenience, we use $\rho_{sum}^{min}$ to represent the sum of $\rho_{i,min}$, i.e., $\rho_{sum}^{min}=\sum_{i=1}^{K}\rho_{i,min}$.  After several mathematical steps, we can obtain \textit{\textbf{Propositions 1}} and \textit{\textbf{2}}. See Appendix A for the proofs of them.

\textit{\textbf{Proposition 1}}: With fixed $r_{min}$, we have
\begin{align}
\rho_{sum}^{min}=(2^{r_{min}}-1)\sum_{i=0}^{K-1}{2^{ir_{min}}}/{\lambda_{K-i}},
\end{align}
and 
\begin{align}
\sum_{i=1}^{K}r_i^U=Kr_{min}+\sum_{i=1}^{K}\triangle r_i^U.
\end{align}
For ease of representation, by defining
\begin{align}
\rho_i^e=(\triangle\rho_i-(2^{r_{min}}-1)\sum_{j=1}^{i-1}\triangle\rho_j)2^{(K-i)r_{min}},
\end{align}
and $n_i^e=(\Psi/\lambda_i+\sum_{j=1}^{i}\rho_{j,min})2^{(K-i)r_{min}},$
we can arrive at $\triangle r_i^U=\log_2\left(1+{\rho_i^e}/{(n_i^e+\sum_{j=1}^{i-1}\rho_j^e)}\right).$

\textit{\textbf{Proposition 2}}: The more power we allocate to the users with stronger channel conditions, the higher the sum rate is. In other words, when all the excess power is allocated to $U_1$, we have the optimal solution as
$\sum_{i=1}^{K}\triangle r_i^U=\triangle r_1^U=\log_2(1+{(\rho_U-\rho_{sum}^{min})\lambda_1}/{(\Psi 2^{Kr_{min}})})$.

Occupying the \textit{\textbf{Propositions}} above, $P_1$ can be derived as
\begin{subequations}
	\begin{align}
	P_2:\quad \mathop {\max\limits_{c_f,\rho_U,\rho_i}}&\quad Kr_{min}+\sum_{i=1}^{K}\triangle r_i^U\nonumber\\
	s.t.
	\quad&(\ref{4c}),(\ref{4f}),(\ref{4g}),(\ref{5a}),\nonumber\\
\label{8a}\quad&Kr_{min}\!\!+\!\!\sum_{i=1}^{K}\triangle r_i^U\leq \frac{R}{\sum_{f=1}^{F}q_f(1-c_f)},\\
	\quad&\triangle r_i^U=\log_2\left(1+\frac{\rho_i^e}{n_i^e+\sum_{j=1}^{i-1}\rho_j^e}\right).
	\end{align}
\end{subequations}
Obviously, $P_2$ is still hard to solve due to  (\ref{4c}), (\ref{5a}),  and (\ref{8a}). If we can fix $\rho_i$, $P_2$ will be facilitated. Therefore, our aim is to find a value of $\rho_U$ which always satisfies  (\ref{4c}) and (\ref{5a}), with any distribution of $\rho_i$.  To elaborate a little further, first, we assume to allocate all excess power to $U_1$ as shown in \textit{\textbf{Proposition 2}}. Obviously, this is the maximum value of the objective function which $\rho_i$ can achieve in various distributions and also the strictest (\ref{8a}) limitation. In this case, (\ref{8a}) can be rewritten as $Kr_{min}+\triangle r_1^U\leq {R}/{(\sum_{f=1}^{F}q_f(1-c_f))}$. Apparently, $\rho_U=0$ is a feasible point, which leads to $\rho_{min}=0$,   $\triangle\rho=0$, and $r_{min}=0$. In this case, (\ref{5a}) and (\ref{8a}) are bound to be satisfied. Consequently, we can achieve $P_3$ as
\begin{subequations}
	\begin{align}
	P_3:\quad \mathop {\max\limits_{c_f, \rho_U}}&\quad\text{obj}=Kr_{min}+\triangle r_1^U\nonumber\\
	s.t.\quad&(\ref{4f}),(\ref{4g}), (\ref{5a}),(\ref{8a}),\nonumber\\
	\label{9e}\quad&\triangle r_1^U=\log_2\left(1+\frac{(\rho_U-\rho_{sum}^{min})\lambda_1}{\Psi 2^{Kr_{min}}}\right),\\
	\label{9f}\quad&\rho_{sum}^{min}=(2^{r_{min}}-1)\sum_{i=0}^{K-1}\frac{2^{ir_{min}}}{\lambda_{K-i}}.
	\end{align}
\end{subequations}
However, $P_3$ is still non-convex. Hence, we divide it into two convex sub-problems to find its optimal solution. For given $c_f$, problem $P_3$ reduces to
\begin{subequations}
	\begin{align}
	P_4:\quad \mathop {\max\limits_{\rho_U}}&\quad \text{obj}\nonumber\\
	s.t.\quad&(\ref{5a}),(\ref{8a}),(\ref{9e}),(\ref{9f}).\nonumber
	\end{align}
\end{subequations}
For given $\rho_U$, problem $P_3$ reduces to
\begin{subequations}
	\begin{align}
	P_5:\quad \mathop {\max\limits_{c_f}}&\quad \text{obj}\nonumber\\
	s.t.\quad&(\ref{4f}),(\ref{5a}),(\ref{8a}),(\ref{9f}).\nonumber
	\end{align}
\end{subequations}
Based on $P_4$ and $P_5$, we can obtain  the \textit{lower bound} of the optimal solution of $P_3$ in {\rm\textbf{Algorithm  \ref{Whole_Algorithm}}}.
\begin{algorithm}[t]
	\caption{Algorithm for Solving The Problem } 
	\label{Whole_Algorithm} 
	\begin{algorithmic}[1] 
		\Require 
		System Parameters, $Converge$=\text{false}, iteration index $l=1$, and tolerance $\delta$.
		\Ensure System sum rate. 
		\While {$Converge$=\text{false}}  
		\State $l=l+1$;
		\State Solve $P_4$ for current $\rho_U^{(l)}$.
		\State With $\rho_U^{(l)}$, solve  $P_5$ for current $c_f^{(l)}$.
		\State Calculate $\text{obj}^{(l)}$.
		\If {$|\text{obj}^{(l)}-\text{obj}^{(l-1)}|\leq \delta$}
	    \State $Converge$=\text{ture}, $\text{obj}^{*}=\text{obj}^{(l)}$.
		\EndIf 
		\EndWhile
		\State \Return{Optimal system sum rate $\text{obj}^{*}$. }
	\end{algorithmic} 
\end{algorithm}
\begin{lemma}
	{\rm\textbf{Algorithm  \ref{Whole_Algorithm}}} guarantees convergence. 
\end{lemma}
\begin{proof}
	Cauchy's theorem proves that function with compact and continuous constraint set
	always converges. Besides, solving $P_4$ and $P_5$ alternatively guarantees the convergence [\ref{convexbook}].\footnote{Since $\rho_U$ and $c_f$ are mutually decoupling, we can calculate them alternatively.} Therefore, proposed algorithm is convergent. 
\end{proof}
\begin{lemma}
	The time complexity of {\rm\textbf{Algorithm  \ref{Whole_Algorithm}}}
	is $\mathcal{O}(\frac{1}{\delta^2})$.
\end{lemma}
\begin{proof}
	The complexity of sub-linear rate, e.g., $f^{(l)}-f^*\leq\delta$ is $\mathcal{O}(\frac{1}{\delta^2})$.
	Therefore, the complexity of the proposed algorithm is obtained. 
\end{proof}
\section{Numerical Results}
In this section, we discuss the performance of the proposed
cache-aided NOMA, and compare it with the  cache-aided OMA systems. The transmit power at RSU is set as $P=10$w and the backhaul capacity constraint is set as $R=5$ bit/s. We consider that RSU serves $K = 2$ and $K=3$ vehicles respectively.  For convenience, we set $(\Omega_1,\Omega_2)=(10,5)$ for the scenario where $K=2$, and  $(\Omega_1,\Omega_2,\Omega_3)=(10,5,1)$ for the scenario where $K=3$. In addition, the detailed settings of the Jakes' model are shown as follows: $v_i$ = 150 km/h, which is practical especially for a highway scenario;  $f_c=5.9$GHz; $\tau=10^{-6}$. The
noise power is set as $\Omega_0=1$w.  As for the CSI estimation errors, we set $\Omega_\epsilon=0.1$. The outage probability threshold
for multicast service is set as $\delta=0.1$.

%
\begin{figure}[t]
	\centering
	\subfigure[$K$=2]{\label{noma_rate_k_2}
		\includegraphics[width=0.5\textwidth]{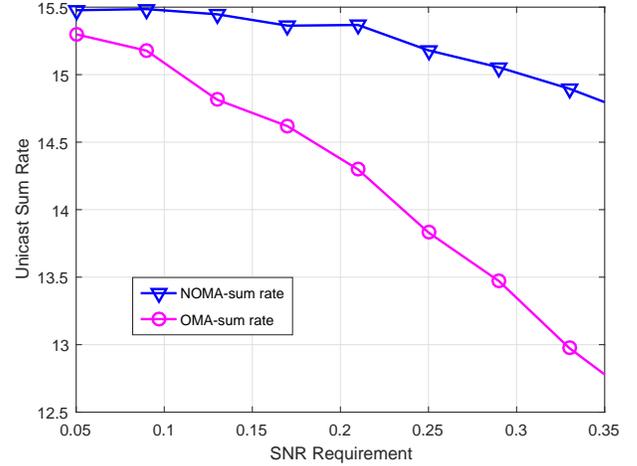}  }
	\subfigure[$K$=3]{\label{noma_rate_k_3}
		\includegraphics[width=0.5\textwidth]{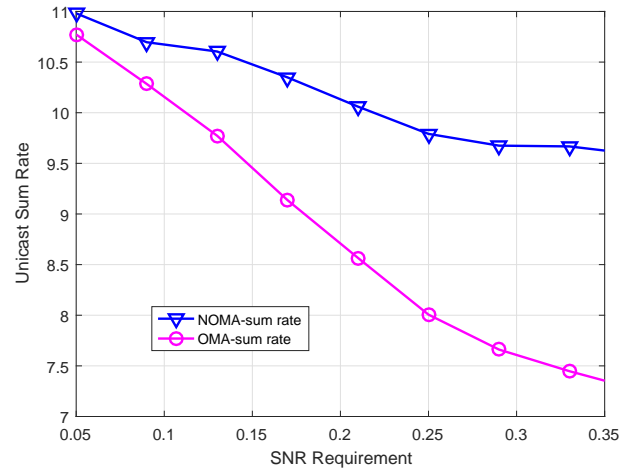} }
	\caption{Sum-rate versus minimum rate constraint with cache size $N=2$.}
	\label{NOMA_rate_constraint}
\end{figure} 

\begin{figure}[t]
	\centering
	\subfigure[$K$=2]{\label{noma_cache_k_2}
	\includegraphics[width=0.5\textwidth]{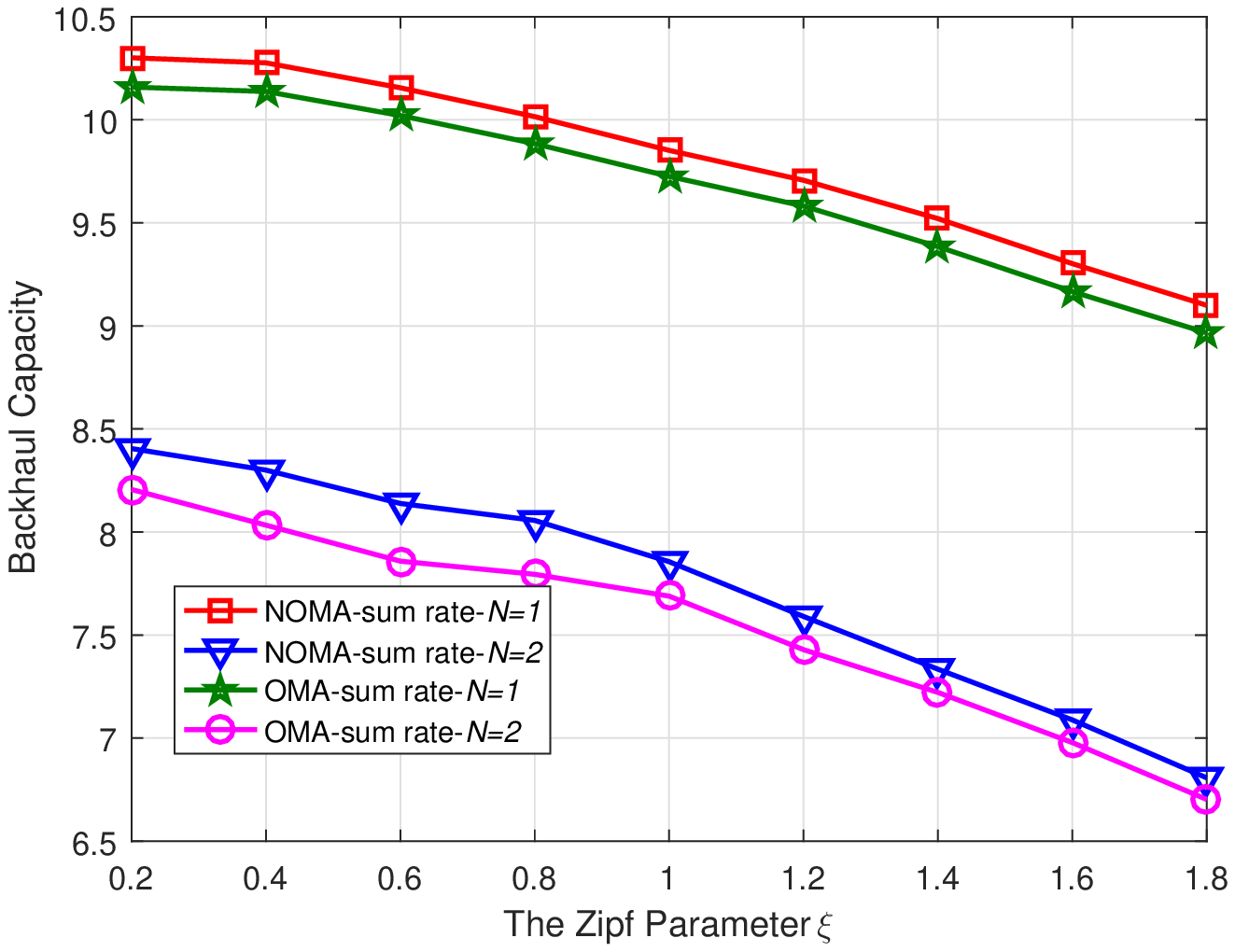} }
\subfigure[$K$=3]{\label{noma_cache_k_3}
	\includegraphics[width=0.5\textwidth]{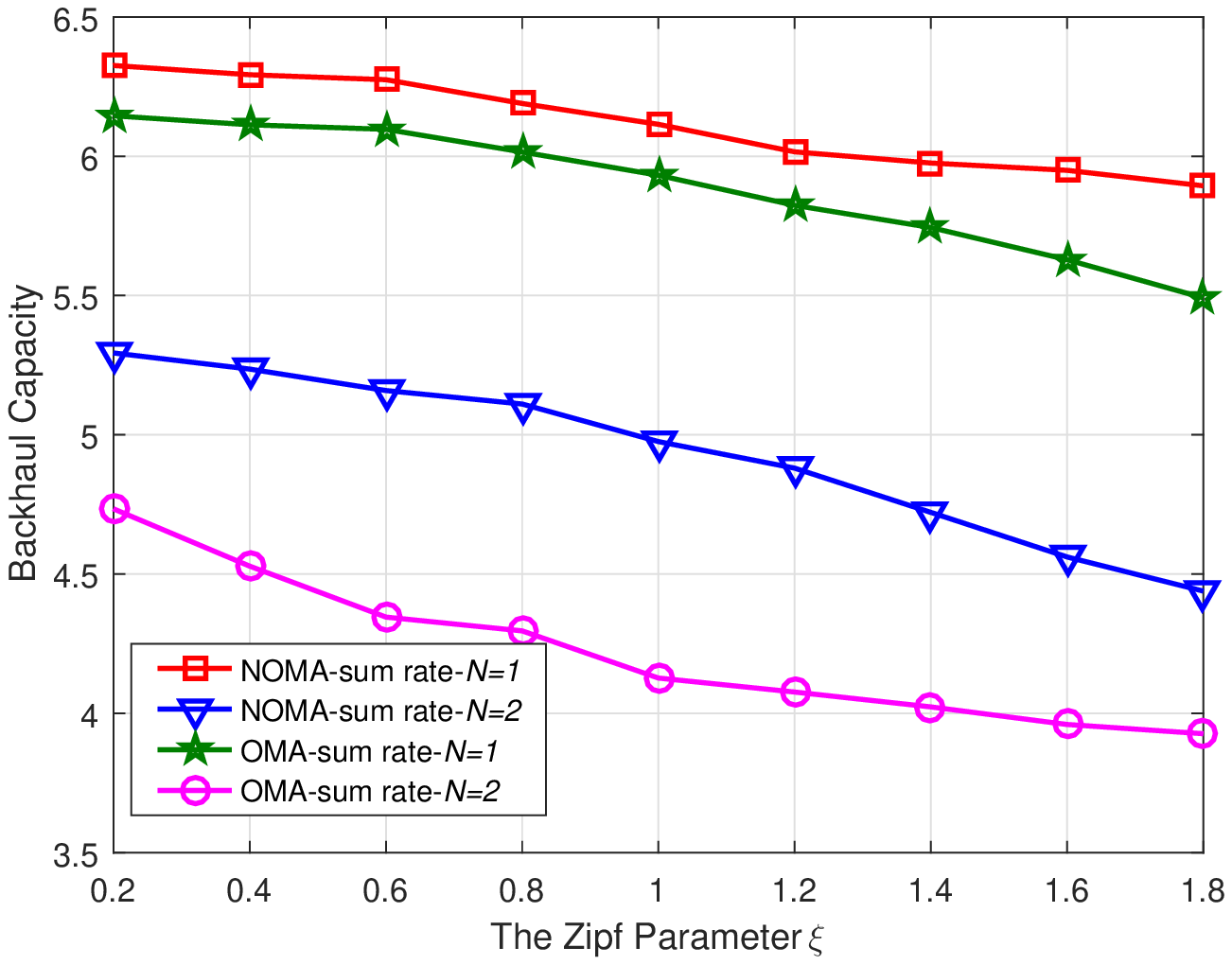} }
\caption{Backhaul capacity versus cache size with minimum rate constraint $r_{min}=0.2$.}
\label{NOMA_cache_size}
\end{figure}

In Fig. \ref{NOMA_rate_constraint}, we compare the unicast sum rate of cache-aided NOMA with that of the OMA counterpart under different minimum rate constraints. As expected, the NOMA scheme outperforms the OMA one in all cases. Obviously, the sum rates decrease  when $r_{min}$ increases, but the decrease is moderate. This is because $Kr_{min}$  is linearly increased while $\triangle r_1^U$ is exponentially decreased.  Furthermore, compare Figs. \ref{noma_rate_k_2} and \ref{noma_rate_k_3}, we can easily find that the systems with three users have lower unicast sum rate. This is because when the transmission power of the RSU is fixed, the increase of the user will also aggravate the interference, which leads to the decrease of the receiving performance, and finally affects the unicast rate.

Figure \ref{NOMA_cache_size} shows the backhaul capacity versus
the zipf parameter $\xi$ for different cache size $N$. Obviously, as $\xi$ increases, the backhaul capacity decreases, which comes from the fact that larger $\xi$ represents the more concentrated request hotspots. In other words, the probabilities that the cached files at RSU are requested by users are larger, which reduces the backhaul overhead.  Moreover, one can observe that  the  backhaul capacity of the NOMA scheme is always larger than that of the OMA one. This is because, compared to OMA, NOMA shows a superior unicast rate performance and therefore requires a relatively higher amount of backhaul resources. Besides, we can find that an increasing number of users will decrease the backhaul capacity, whose cause is the same as that of the previous figure.
\section{Conclusions}
In this paper, we have incorporated multicast and unicast
services into a cache-aided SISO vehicular NOMA system with high   mobility. We have formulated
an optimization problem to maximize the unicast sum rate  subject to the peak power, the backhaul capacity, the minimum unicast rate, and the maximum multicast outage probability constraints. The proposed non-convex problem has been appropriately solved
by the proposed lower bound relaxation method. Simulation results
have demonstrated that our proposed cache-aided NOMA scheme  outperforms the OMA counterpart.

\begin{appendices}
\setcounter{equation}{0}
\renewcommand{\theequation}{\thesection.\arabic{equation}}
\section{Proofs of \textit{\textbf{Propositions 1}} and \textit{\textbf{2}}}
Being allocated $\rho_{i,min}$, the unicast rate of $U_i$ can achieve $r_{min}$, i.e.,
\begin{align}\label{definition_of_rmin}
r_{min}=\log_2\left(1+\rho_{i,min}\lambda_i/(\sum_{j=1}^{i-1}\rho_{j,min}\lambda_i+\Psi)\right),
\end{align}
which yields
\begin{align}\label{Initial_formula}
2^{r_{min}}-1=\rho_{i,min}/(\sum_{j=1}^{i-1}\rho_{j,min}+\Psi/\lambda_i).
\end{align}
Using partition ratio theorem, (\ref{Initial_formula}) can be formulated as 
\begin{align}\label{partition_ratio_theorem}
\frac{(2^{r_{min}}-1)\sum_{j=1}^{i-1}\triangle\rho_{j}}{\sum_{j=1}^{i-1}\triangle\rho_{j}}=\frac{\rho_{i,min}}{\sum_{j=1}^{i-1}\rho_{j,min}+\Psi/\lambda_i}\nonumber\\
=\frac{\rho_{i,min}+(2^{r_{min}}-1)\sum_{j=1}^{i-1}\triangle\rho_{j}}{\sum_{j=1}^{i-1}\rho_{j,min}+\sum_{j=1}^{i-1}\triangle\rho_{j}+\Psi/\lambda_i}.
\end{align}
Substituting (\ref{partition_ratio_theorem}) into (\ref{definition_of_rmin}), we can obtain
\begin{align}
r_{min}=\log_2\left(1+\frac{\rho_{i,min}+(2^{r_{min}}-1)\sum_{j=1}^{i-1}\triangle\rho_{j}}{\sum_{j=1}^{i-1}\rho_{j}+\Psi/\lambda_i}\right).
\end{align}
Therefore, $\triangle r_i^U$ can be expressed as
\begin{align}
\triangle r_i^U&=r_i^U-r_{min}\nonumber\\
=&\log_2\left(1+\frac{\triangle\rho_{i}-(2^{r_{min}}-1)\sum_{j=1}^{i-1}\triangle\rho_{j}}{\sum_{j=1}^{i}\rho_{j,min}+\frac{\Psi}{\lambda_i}+2^{r_{min}}\sum_{j=1}^{i-1}\triangle\rho_{j}}\right)\nonumber\\
=&\log_2(1+\frac{\mathscr{P}_i}{\mathscr{N}_i+\mathscr{Q}_i}),
\end{align}
where $\mathscr{P}_i=\triangle\rho_{i}-(2^{r_{min}}-1)\sum_{j=1}^{i-1}\triangle\rho_{j}$, $\mathscr{N}_i=\sum_{j=1}^{i}\rho_{j,min}+\frac{\Psi}{\lambda_i}$, and $\mathscr{Q}_i=2^{r_{min}}\sum_{j=1}^{i-1}\triangle\rho_{j}$. Using the properties of \emph{recurrence}, we have 
\begin{align}\label{QandP}
\mathscr{Q}_i=2^{r_{min}}\sum_{j=1}^{i-1}\triangle\rho_{j}=\sum_{j=1}^{i-1}(2^{r_{min}})^{i-j}\mathscr{P}_j.
\end{align}
Let $\rho_i^e$ denote $\mathscr{P}_i(2^{r_{min}})^{K-i}$. Then, we can rewrite (\ref{QandP}) into
$
\mathscr{Q}_i(2^{r_{min}})^{K-i}=\sum_{j=1}^{i-1}(2^{r_{min}})^{K-j}\mathscr{P}_j=\sum_{j=1}^{i-1}\rho_j^e.
$
Therefore, we can derive
\begin{align}\label{finverofriu}
\triangle r_i^U=\log_2(1+\frac{\rho_{i}^e}{n_i^e+\sum_{j=1}^{i-1}\rho_j^e}),
\end{align} 
where $n_i^e=(\Psi/\lambda_i+\sum_{j=1}^{i}\rho_{j,min})2^{(K-i)r_{min}}.$  On the other hand, (\ref{Initial_formula}) can be rewritten as $\rho_{i,min}=(2^{r_{min}}-1)(\sum_{j=1}^{i-1}\rho_{j,min}+\Psi/\lambda_i).$
After the recurrence operation, we have 
\begin{align}
\rho_{i,min}=\frac{2^{r_{min}}-1}{\lambda_i}+\frac{(2^{r_{min}}-1)^22^{(i-j-1)r_{min}}}{\lambda_i},
\end{align}
which results in 
$\rho_{sum}^{min}=\sum_{i=1}^{K}\rho_{i,min}=(2^{r_{min}}-1)\sum_{i=0}^{K-1}{2^{ir_{min}}}/{\lambda_{K-i}}.$ Because $\rho_{sum}^{min}$ represents all the excess power, $\rho_{i}^e\le\rho_{sum}^{min}$. Therefore, when $i=1$, $\rho_{i}^e=\rho_{sum}^{min}$, (\ref{finverofriu}) achieves its optimal value. The proofs complete. 
\end{appendices}
\setcounter{equation}{0}
\renewcommand{\theequation}{\thesection.\arabic{equation}}
\section{Supplementary Material}
\subsection{The Detailed Derivations of (1) and (2)}
As we know, the transmit signal at RSU is
\begin{align}
	{z}=\sqrt{\beta_MP}{x}_M+\sum_{i=1}^{K}\sqrt{\beta_UP_i}{x}_i, 
\end{align}
where $\text{E}[\vert x_M\vert^2]=\text{E}[\vert x_i\vert^2]=1$.
Then the received signal at user $i$ can be derived as
\begin{align}
	{y}_i=\sqrt{\beta_MP}{h}_i{x}_M+\sum_{j=1}^{K}\sqrt{\beta_UP_j}{h}_i{x}_j+n_i,
\end{align}
which can be rewritten as 
\begin{align}
	{y}_i&=\sqrt{\beta_MP}\sqrt{1-\phi^2}\hat{{{h}}}_i{x}_M\nonumber\\
	&\overbrace{+\sqrt{\beta_MP}\phi {\epsilon}_i{x}_M
		+\sum_{j=1}^{K}\sqrt{\beta_UP_j}(\sqrt{1-\phi^2}\hat{{{h}}}_ia+\phi {\epsilon}_i){x}_j+n_i}^{interference}.
\end{align}
Without loss of generality, multicast message always has a higher priority than the unicast one. Therefore, the receiver should first decode the multicast message ($x_M$) and subtract it from $y_i$. In this way, the SINR of $x_M$ at user $i$ can be obtained by
\begin{align}\hspace*{-1cm}
	\gamma_i^M&={\beta_MP(1-\phi^2)\vert\hat{{{h}}}_i\vert^2}\nonumber\\
	&\div\left(\beta_MP\phi^2\vert{\epsilon}_i\vert^2+\sum_{j=1}^{K}\beta_UP_j(1-\phi^2)\vert\hat{{{h}}}_i\vert^2\right.\nonumber\\
	&\left.+\sum_{j=1}^{K}\beta_UP_j\phi^2\vert{\epsilon}_i\vert^2+\Omega_0\right),
\end{align}
which equals to the SINR in (1). After decoding $x_M$, user $i$ aims to obtain $x_i$ from the superposed signal 
\begin{align}
	{y}_i=\sum_{j=1}^{K}\sqrt{\beta_UP_j}{h}_i{x}_j+n_i.
\end{align}
Recall $\vert h_1\vert^2\geq\cdots\geq\vert h_i\vert^2\geq\cdots\geq\vert h_K\vert^2$, user i first decodes the data symbols for the users with weaker channels, subtract them through SIC technique, and then decoding the data symbol for itself. Consequently, we can obtain 
\begin{align}
	{y}_i&=\sqrt{\beta_UP_i}\sqrt{1-\phi^2}\hat{{{h}}}_i{x}_i\nonumber\\
	&+\overbrace{
		\sum_{j=1}^{i-1}\sqrt{\beta_UP_j}(\sqrt{1-\phi^2}\hat{{{h}}}_i+\phi {\epsilon}_i){x}_j+\sqrt{\beta_UP_i}\phi {\epsilon}_i{x}_i+n_i}^{interference},
\end{align}
and
\begin{align}\hspace*{-0.5cm}
	\gamma_i^U=\frac{{\beta_UP_i}{(1-\phi^2)}\vert\hat{{{h}}}_i\vert^2}{\sum_{j=1}^{i-1}{\beta_UP_j}{(1-\phi^2)}\vert\hat{{{h}}}_i\vert^2+\sum_{j=1}^{i}{\beta_UP_j}\phi^2\vert{\epsilon}_i\vert^2+\Omega_0}.
\end{align}
In this way, we can finally derive (1) and (2).
\subsection{The Derivation of (6)}
Recall the instantaneous rate of $x_i$ observed at $U_i$, i.e.,
\begin{align}
	r_i^U=\log_2\left(1+\rho_i\lambda_i/(\sum_{j=1}^{i-1}\rho_j\lambda_i+\overbrace{\sum_{j=1}^{i}\rho_jb+a}^{relax})\right).
\end{align}
Since the last two parts in the denominator are hardly handled, we herein use the lower bound relaxation method and replace them by a constant, i.e.,
\begin{align}
	\sum_{j=1}^{i}\rho_jb+a=\sum_{j=1}^{i}\frac{\beta_UP_i}{\Omega_0}b+a\underset{i\rightarrow K}{\overset{\beta_U\rightarrow 1}{\rightarrow}}\sum_{j=1}^{K}\rho b+a=\Psi. 
\end{align}
In this way, we can derive (6).

\end{document}